\documentclass[draftclsnofoot]{IEEEtran}
 \onecolumn
\usepackage{latexsym}
\usepackage{cite}
\usepackage{bm}
\usepackage{amsmath, amssymb}
\usepackage{graphicx}
\usepackage{epsfig}
\usepackage{psfrag}
\usepackage{dsfont}
\usepackage{setspace}
\usepackage{color}
\usepackage{slashbox}

\usepackage{theorem}
\theoremheaderfont{\bfseries\upshape} \theoremstyle{plain}

\theorembodyfont{\rmfamily}

\theoremheaderfont{\itshape}
 
\newtheorem{example}{Example}
\newtheorem{definition}{Definition}

\theorembodyfont{\itshape}
\newtheorem{theorem}{Theorem}

\newtheorem{lemma}{Lemma} 

\def\leq{\leqslant}
\def\geq{\geqslant}

\def\cS{{\mathcal{S}}}
\def\cS{{\mathcal{S}}}
\def\bbZ{{\mathbb{Z}}}
\def\bbR{{\mathbb{R}}}
\def\cX{{\mathcal{X}}}
\def\cY{{\mathcal{Y}}}
\def\beq{\begin{equation}\begin{split}}
\def\eeq{\end{split}\end{equation}}
\begin{document}

\title{\LARGE Zero-error feedback capacity via dynamic programming}
%

\author{Lei Zhao and Haim Permuter
\thanks{Lei Zhao is with the Department of Electrical Engineering,
 Stanford University, Stanford, CA 94305. Haim Permuter is with Department of Electrical \& Computer
 Engineering Department,  Ben-Gurion University of the Negev, Beer-Sheva, Israel.
Authors' emails: leiz@stanford.edu and haimp@bgu.ac.il.}
}%

\maketitle

\begin{abstract}
In this paper, we study the zero-error capacity for finite state
channels with feedback when channel state information is known to
both the transmitter and the receiver. We prove that the zero-error
capacity in this case can be obtained through the solution of  a
dynamic programming problem. Each iteration of the dynamic
programming provides lower and upper bounds on the zero-error
capacity, and in the limit, the lower bound coincides with the
zero-error feedback capacity. Furthermore, a sufficient condition
for solving the dynamic programming problem is provided through a
fixed-point equation. Analytical solutions for several examples are
provided.
\end{abstract}

\begin{keywords}
Bellman equations, competitive Markov decision processes, dynamic
programming, feedback capacity, fixed-point equation,
infinite-horizon average reward, stochastic games, zero-error
capacity.
\end{keywords}

\section{Introduction}

In 1956, Shannon \cite{shannon56} introduced the concept of
zero-error communication, which requires that the probability of
error in decoding any message transmitted through the channel to be
zero. 
Although the zero-error capacity for general channels remains an
unsolved problem (see \cite{KornerOrlitski98Zero-error} for a
comprehensive survey of zero-error information theory), Shannon
\cite{shannon56} showed that for discrete memoryless channels (DMC)
with feedback the zero-error capacity is either zero (if any two
inputs can generate a common output) or equal to:
\begin{equation}\label{e_memoryless_zero}
C_0^{FB}=\max_{P_X} \log_2 \left[ \max_{y} \sum_{x\in G(y)} P_X(x)
\right]^{-1},
\end{equation}
where $P_X$ is the channel input distribution, $y$ is an output
realization of the channel, and $G(y)$ is the set of inputs that
have a positive probability of generating the output $y$, i.e.,
$G(y)\triangleq \{x: P_{Y|X}(y|x)>0 \}$. The achievability proof of
(\ref{e_memoryless_zero}) is based on a determinist scheme rather
than on a random coding scheme, as used for showing the
achievability of regular capacity.

In this paper, we study the zero-error feedback capacity for finite
state channels (FSC), a family of channels with memory. We make the
assumptions that channel state information (CSI) is available both
to the transmitter and to the receiver. In this case, we solve the
zero-error capacity that depends only on the topological properties
of the channel. A similar setup has been used by Chen and Berger
\cite{Chen05}, who solved the regular channel capacity by finding
the optimal stationary and nonstationary input processes that
maximize the long-term directed mutual information. In
\cite{Ahl_kaspi87} and \cite{PermuterCuffVanRoyWeissman07_Chemical},
the zero-error capacity of the chemical channel with feedback was
derived. The chemical channel is a special case of FSCs. With
feedback, the transmitter knows the state of the chemical channel
while the receiver does not, which is different from our setup.
Other related work can be found in \cite{NayaRose05}, which
addresses the zero-error capacity for compound channels.

The remaining of the paper is organized as follows. In Section
\ref{s_channel_model}, we introduce the channel model and the
dynamic programming problem formulation. In Section
\ref{s_sufficient_C=0}, we use a finite-horizon dynamic programming
(DP)  to provide a condition for the channel to have zero zero-error
capacity. In Section \ref{s_DP_associated}, we define an
infinite-horizon average reward DP problem and link its solution the
the zero-error capacity. In Sections \ref{sec.converse} and
\ref{sec.direct}, we prove the converse and direct parts
respectively. In Section \ref{s_seolving_DP}, we explain how to
evaluate the infinite-horizon average reward DP; in particular, we
provide a sequence of lower and upper bounds that are easy to
compute and prove the Bellman equation theorem for the particular
DP, namely, a fixed-point equation that is a sufficient condition
for verifying the optimality of a solution. In Section
\ref{s_examples}, we evaluate and then find analytically the zero
error feedback capacity of several examples.

\section{Channel Model and Preliminaries}\label{s_channel_model}
We use calligraphic letter $\cX$ to denote the alphabet and $|\cX|$
to denote the cardinality of the alphabet. Subscripts and
superscripts are used to denote vectors in the following way:
$x^j=(x_1,...,x_j)$ and $x_i^j=(x_i,...,x_j)$ for $i\leq j$. Next we
introduce the channel model and the DP formulation.

\subsection{Channel model and zero-error capacity definition}
An FSC\cite[ch. 4]{Gallager68} is a channel that, at each time
index, has a state whic belongs to a finite set $\cS$ and has the
property that, given the current input and state, the output and the
next state is independent of the past inputs, outputs and states,
i.e.,
\begin{equation}
 p(y_t,s_{t+1}|x_1^t,s_1^t)=p(y_t,s_{t+1}|x_t,s_t).
\end{equation}
For simplicity, we assume that the channel has the same input
alphabet $\cX$ and the same output alphabet $\cY$ for all states.
The alphabets $\cX$ and $\cY$ are both finite. Without loss of
generality, we can assume that $\cX=\{1,2,...,|\cX|\}$. We consider
the communication setting shown in Fig.~\ref{fig.ch}, where the
state of the channel is known to the encoder and to the decoder.

\begin{figure}{
 \psfrag{M}[][][1]{Message}
 \psfrag{m}[][][1]{$m$}
 \psfrag{E}[][][1]{Encoder}
 \psfrag{A}[][][1]{$\;\;x_i(m,y^{i-1},s^i)$}
 \psfrag{x}[][][1]{$x_i$}
 \psfrag{C}[][][1]{FSC}
 \psfrag{P}[][][1]{\ \ \;\;$p(y_i,s_{i+1}|x_i,s_i)$}
 \psfrag{y}[][][1]{$y_i,s_{i+1}\;$}
 \psfrag{D}[][][1]{Decoder}
 \psfrag{MH}[][][1]{$\;\;\hat m(y^n,s^{n+1})$}
 \psfrag{mh}[][][1]{$\hat{m}$}
 \psfrag{UD}[][][0.9]{Unit Delay}
 \psfrag{yd}[][][1]{$y_{i-1},s_{i}$}
\centerline{\includegraphics[width=13cm]{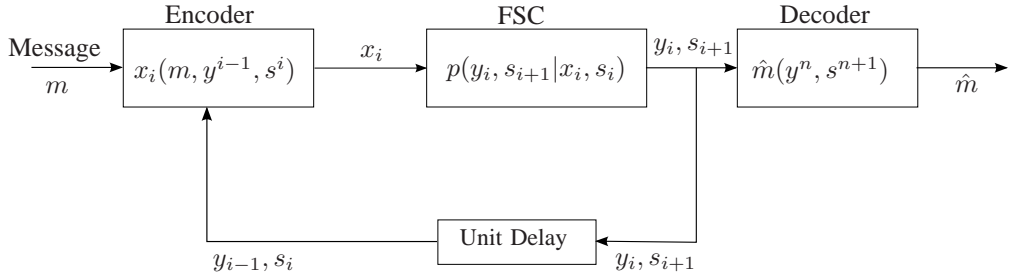}}
\caption{Communication model: a finite state channel (FSC) with
feedback and state information at the decoder and encoder.
}\label{fig.ch} }\end{figure}

An $(M,n)$ zero-error feedback code of length $n$ is defined as a
sequence of encoding mappings $x_t(m,y^{t-1},s^t)$ and a decoding
function $\hat m= g(y^n,s^{n+1})$, where a message $m$ is selected
from a set $\{1,...,M\}$. The probability of error is required to be
zero, i.e., $\Pr\left\{ g(Y^n,S^{n+1}) \neq m|\mbox{message }m
\mbox{ is sent}\right\}=0$ for all messages $m\in\{1,2,...,M\}$. We
emphasize that the size of the message set $M$ does not depend on
the initial state of the channel; hence, the probability of error
decoding needs to be zero for any initial state.

\begin{definition}\label{d_achievable_rate}
A rate $R$ is {\it achievable} if there exists an $(M,n)$ zero-error
feedback code such that $R\leq\frac{\log_2M}{n}$.
\end{definition}

\begin{definition}
The {\it operational zero-error capacity} of an FSC is defined as the supreme of all achievable rates.
\end{definition}

Throughout this paper we use the following alternative and
equivalent definition of the operational zero-error capacity.
\begin{definition}\label{def_zero_cap}
Let $M(n,s)$ be the maximum number of messages that can be
transmitted with zero error in $n$ transmissions when the initial
state of the channel is $s\in\cS$. Define
\begin{equation}
a_n=\min_{s\in\cS} \log_2
M(n,s).
\end{equation}
The {\it operational zero-error capacity} is given by:
\begin{equation}\begin{split}\label{eq.c0f}
 C_{0}\triangleq &\lim_{n\rightarrow\infty}\frac{a_n}{n}
                =\lim_{n\rightarrow\infty}\frac{\min_{s\in\cS}\log_2 M(n,s)}{n},
\end{split}\end{equation}
where the limit is shown to exist.
\end{definition}

Since the transmitter knows the state, the sequence
$\{a_n\}$ is super additive, i.e., \mbox{$a_{n+m}\geq a_n+a_m$} and
$\frac{a_n}{n}\leq |\cX|$. By Fekete's lemma \cite[Ch.
2.6]{CombinatorialOptimization_SchrijverBook}, $\lim\frac{a_n}{n}$
 exists and is equal to $\sup\frac{a_n}{n}$. Note that $R\leq
\lim\frac{a_n}{n}$ holds for any achievable rate $R$, and any rate
less than $\lim\frac{a_n}{n}$ is achievable, which are simple
consequences of Definition \ref{d_achievable_rate}. Thus,
$\lim\frac{a_n}{n}$ defines the zero-error capacity.

\subsection{Dynamic programming}
For the standard Markov decision process (MDP), we have the dynamic
programming equation
\cite{Arapos93_average_cose_survey,Bertsekas05}:
\begin{equation}\label{e_standard}
U_n(s)=\max_{a\in A(s)} \left\{r(s,a)+\sum_{s'=1}^N
P(s'|s,a)U_{n-1}(s') \right\},
\end{equation}
where $r(s,a)$ is the reward, given that we are at state $s\in
\mathcal S$, and we perform action $a\in \mathcal A$. The term
$U_n(s)$ is the total reward after $n$ steps (a.k.a.  the
"reward-to-go" in $n$ steps) when we start at time $s$. The
conditional distribution $P(s'|s,a)$ is the probability of the next
state $s'\in \mathcal S$, given the current state $s\in \mathcal S$
and action $a\in \mathcal A$.

The dynamic programming equation that is associated in this paper
with the zero-error capacity has the form
\begin{equation}\label{e_nonstandard}
U_n(s)=\max_{a\in A(s)} \min_{s'\in \mathcal S(a)}
\left\{r(s,a,s')+U_{n-1}(s') \right\},
\end{equation}
where $r(s,a,s')$ is the reward, given the current state $s$, the
action $a$ and the next state $s'$. The reward may be any real
number, including $\pm \infty.$ The value $U_n(s)$ is defined as
before, i.e., the total reward in $n$ steps when starting at state
$s$.

The DP equation in (\ref{e_nonstandard}) may be viewed as a
stochastic game \cite{Shapley1953StochasticGame}, which is a.k.a
competitive MDP \cite{Filar97CompetitveMDP}, in which there are two
asymmetric players. Player 1, the leader, takes an action
$a\in\mathcal A(s)$, which may depend on the current state and
Player 2, the follower, determines the next state $s'\in\mathcal S$.
Player 2 sees the state of the game $s$ and the action of player 1.
In the zero-error capacity problem, Player 1 would be the user who
designs the code to maximize the transmitted rate, and Player 2
would be Nature, which chooses the  next state  to minimize the
transmitted rate.

\section{A sufficient and necessary condition for
$C_0=0$}\label{s_sufficient_C=0} Shannon \cite{shannon56} showed
that for a DMC, which is an FSC with only one state, if any two
input letters have at least one common output, it is impossible to
distinguish between two messages with zero-error. Using
finite-horizon dynamic programming, we derive in this section a
sufficient and necessary condition for an FSC to have $C_0=0$, i.e.,
the zero-error capacity is zero.

%

\begin{definition}
Two input letters $x_1$ and $x_2$ are called {\it adjacent} at state $s$ if there exists an output letter $y$ and a state $s'$ such that $p(y,s'|x_1,s)>0$ and $p(y,s'|x_2,s)>0$.
 \end{definition}
 \begin{definition}
 A state $s$ is {\it positive} if there exist two input letters that are not {\it adjacent} at state $s$.
 \end{definition}
 The intuition behind the result in this section is that if the
channel  undergoes only non-{\it positive} states during the
transmission, we cannot distinguish between two messages based on
the output sequence and the channel state sequence, since they could
result from either message.

To determine whether $C_0=0$, we form the following dynamic
programming equation,
\begin{equation}\label{eq.cadd}
V_n(s)=r(s)+\max_{x\in \cX}\min_{s'\in \mathcal S(s,x)}V_{n-1}(s'),
\end{equation}
where $V_0(s)=0,\forall s\in\cS$, $\cS(s,x)=\{s': p(y,s'|x,s)>0
\text{ for some } y\in \cY  \}$, and reward $r(s)=1$ if state $s$ is
{\it positive}, while $r(s)=0$ if state $s$ is not {\it positive}.
\begin{lemma} \label{l_mono}(\it monotonicity of $V_n(s)$.)
The total reward $V_n(s)$ is non-negative and    non-decreasing in
$n$, i.e.,
\begin{eqnarray}
0\leq V_n(s)\leq V_{n+1}(s), \forall n=1,2,3,...  \mbox{ and } s\in
\mathcal S.
\end{eqnarray}
\end{lemma}
\begin{proof}
Let $\tilde V_n(s)=r(s)+\max_{x\in \cX}\min_{s'\in \mathcal
S(s,x)}\tilde V_{n-1}(s')$ and $\tilde V_0(s)\geq V_0(s),\; \forall
s\in \mathcal S.$ Then, by induction, we have $\tilde V_n(s)\geq
V_n(s),\; \forall n=1,2,3,...  \mbox{ and } s\in \mathcal S .$ Since
$r(s)\geq0, \forall s\in \mathcal S$, then $V_1(s)\geq 0$. Let us
define $\tilde V_0(s)=V_1(s)$. Since $\tilde V_0(s)\geq V_0(s)$, we
obtain that $\tilde V_n(s)\geq V_n(s)$, which means that
$V_{n+1}(s)\geq V_n(s)\geq 0$ .
\end{proof}

\begin{table}[h]
\caption{Interpretation of The DP given in (\ref{eq.cadd}),  which
corresponds to determining whether $C_0>0$.} \centering
\label{ta_summary_dp}

\begin{tabular}{|l|l|}

\hline  The DP given in (\ref{eq.cadd})& Interpretation of the DP\\
\hline \hline

state $s$ of the DP &  state $s$ of the channel\\
\hline
reward $r(s)$=1 & state $s$ is positive; at least one bit can be transmitted error-free\\
\hline
reward $r(s)$=0 & state $s$ is non positive; no bits can be transmitted error-free\\
\hline

Player 1 takes action $x$ in order to & encoder chooses input $x$ in order to\\
 maximize the reward of the DP& maximize the number of positive states visited\\ \hline
Player 2 chooses next state in order to&  Nature chooses next state and output to\\
minimize the reward of the DP& minimize the number of messages
transmitted\\\hline

\hline

$V_n(s)$- total reward in $n$ rounds, & number of positive states visited in $n$ \\
starting the game from state $s$ & usages of the channel starting at state $s$,\\
\hline \hline

\end{tabular}
\end{table}

This DP can be viewed as a two-person game, where $V_n(s)$ is the
game result after $n$ steps starting with initial state  $s$. Player
1 chooses the input letter $x$, and Player 2 chooses the next state
$s'$. Both players know the current state $s$, and the reward of the
game is a function only of the current state only, i.e., $r(s)$.
Player 1 makes the first play, and the two players make alternative
plays thereafter. The goal of Player 1 is to maximize the number of
times the channel visits a {\it positive} state, and Player 2 tries
to minimize it. The interpretation of the DP as a stochastic game
between the user and Nature is summarized in Table
\ref{ta_summary_dp}.

The following lemma states that if the total reward of the
stochastic game is zero after $n$ rounds with initial state $s$,
i.e., $V_n(s)=0$, then only one message can be sent error-free
through $n$ uses of the channel with initial state $s$.

\begin{lemma}\label{lemma.vc}
$V_n(s)=0$ implies $M(n,s)=1$ and $V_n(s)>0$ implies $M(n,s)>1.$
\end{lemma}
\begin{proof}
First, we observe that so as to send two or more  messages in $n$
uses of the channel, a positive state should be visited with
probability one. Once a positive state is visited, we can use two
inputs that are not adjacent to transmit without error one bit (two
messages). If a positive state is not visited, then there are no two
inputs that can distinguish between two messages.

The stochastic game given in (\ref{eq.cadd}) verifies  whether a
positive state is visited with probability 1.  In the stochastic
game, the rewards $r(s)=1$  and $r(s)=0$ indicate that state $s$ is
positive and non-positive, respectively. Player 1 is the encoder
which wants to visit a positive state and Player 2 is Nature which
chooses the output and the state such that a positive state will not
be visited. A total reward $V_n(s)=0$ implies that in $n$
transmissions with initial state $s$, with positive probability, the
channel undergoes only non-{positive} states, regardless of the
inputs. Thus $V_n(s)=0$ implies $M(n,s)=1$.
\end{proof}

According to Lemma \ref{l_mono}, $V_n(s)$ is non-negative and
non-decreasing in $n$ for any $s\in\mathcal S$. Thus,
$\min_{s\in\cS} V_n(s)$ is also nondecreasing in $n$, and therefore
$\lim_{n\rightarrow\infty}\min_{s\in\cS} V_n(s)$ is well defined (it
may also be infinite). If $\lim_{n\rightarrow\infty}\min_{s\in\cS}
V_n(s)=0$, then $\min_{s\in\cS} V_n(s)=0, \forall n$ and invoking
Lemma \ref{lemma.vc}, $ \min_{s\in\cS}M(n,s)=1$, which gives $C_0=0$
by definition.
%
The next lemma states that to verify whether
$\lim_{n\rightarrow\infty}\min_{s\in\cS} V_n(s)>0$, it is enough to
calculate a finite-horizon problem.

\begin{lemma}\label{lemma.v}
  \begin{equation*}
  \lim_{n\rightarrow\infty}\min_{s\in\cS} V_n(s)=0 \Longleftrightarrow \min_{s\in\cS} V_{|\cS|}(s)=0
  \end{equation*}
\end{lemma}

\begin{proof}
The $\Longrightarrow$ direction follows from Lemma \ref{l_mono},
which states that for any $s\in \mathcal S$,  $V_n(s)$ is a
non-negative and non-decreasing function in $n$.

Now we prove the $\Longleftarrow$ direction. 
Define $\cS_n$, the set of initial states for which the reward is
zero after $n$ rounds of the stochastic game, i.e.,
$\cS_n=\{s\in\cS: V_n(s)=0\}$. Note that
$\cS_{n+1}\subseteq\cS_{n}$, $\cS_0=\cS$ and $\cS_1=\{s\in\cS:
r(s)=0\}$.


First, we claim that there exists $n^*$, $0\leq n^* \leq |\cS|-1$,
for which $\cS_{n^*}=\cS_{n^*+1}$ must hold, where $\cS_{n^*}$ is
non-empty. Otherwise $\cS_{n+1}$ has at least one less element than
$\cS_{n}$ for $0\leq n \leq |\cS|-1$, and therefore
$\cS_{|S|}=\emptyset$. If $\cS_{|S|}$ is empty, it means that
$\min_{s\in\cS} V_{|\cS|}(s)>0$, which contradicts  our assumption.

The equality between $\cS_{n^*}$ and $\cS_{n^*+1}$ means that when
the channel starts at some $s\in \cS_{n^*+1}$, for any input letter
$x$, there exists an action of Player 2 such that the next state
$s'$ would satisfy $s'\in \cS_{n^*}$. Define this strategy of Player
2 as a function $A_2(\cdot,\cdot):\cS_{n^*} \times \mathcal X
\mapsto \cS_{n^*}$, namely, given $s\in \mathcal S_{n^*}$, and any
input  $x\in \mathcal X$,  the next step $s'$ depends on $s$ and $x$
by the function $A_2(s,x)$ such that $s'=A_2(s,x)$. We claim that
$\cS_{n^*+k}=\cS_{n^*},\forall k\geq0$, i.e., once the set $\cS_n$
stops shrinking, it will stay the same. To prove this, let us fix an
arbitrary $s\in \cS_{n^*+1}$. Since $S_1\subseteq S_n$,  $s\in
\cS_1$ and $r(s)=0$. We have
\begin{equation}\begin{split}
V_{n^*+2}(s)&=r(s)+\max_{x\in \cX}\min_{s'\in \mathcal S(s,x)}V_{n^*+1}(s')\\
          &=\max_{x\in \cX}\min_{s'\in \mathcal S(s,x)}V_{n^*+1}(s')\\
      &\leq \max_{x\in \cX}V_{n^*+1}(A_2(s,x))\\
      &=0
\end{split}
\end{equation}
Therefore $\cS_{n^*+2}=\cS_{n^*+1}$. Repeating the same argument, we
have $\cS_{n^*+k}=\cS_{n^*}, \forall k>0$, which means that
$V_n(s^*)=0, \forall n$. This completes the proof.
\end{proof}

The following theorem state the necessary and sufficient condition for $C_0=0$ through the stochastic game.
\begin{theorem}\label{theorem.C=0}
The zero -error capacity is positive if and only if the total reward
$\min_{s\in\cS} V_{|\cS|}(s)$ is positive, i.e.,
\begin{equation}
\min_{s\in\cS} V_{|\cS|}(s)=0 \Longleftrightarrow C_0=0.
\end{equation}
\end{theorem}
\begin{proof}

If $\min_{s\in\cS} V_{|\cS|}(s)=0$, then according to Lemma
\ref{lemma.v} $\lim_{n\rightarrow\infty}\min_{s\in\cS} V_n(s)=0$,
and following Lemma \ref{lemma.vc} it follows that $\min_s M(n,s)=1$
for any $n$; hence $C_0=0$.

If $\min_{s\in\cS} V_{|\cS|}(s)>0$, then in according to Lemma
\ref{lemma.vc},  $\min_{s\in\cS} M(|\cS|,s)\geq 2$, and following
from the definition of zero-error capacity $C_0\geq
\frac{1}{|\mathcal S|}.$
\end{proof}
%
%
%

\section{The Dynamic Programming Problem associated with the
Channel}\label{s_DP_associated} In this section, we define a dynamic
programming problem associated with the channel. The solution to the
problem is later used to determine the feedback capacity of the
channel.

Denote $G(y,s^\prime|s)=\{x: x\in \cX, p(y,s^\prime|x,s)>0\}$, i.e.,
$G(y,s^\prime|s)$ is the set of input letters at state $s$ that can
drive the channel state to $s^\prime$ while yielding an output
letter $y$ with positive probability. Denote $W(\cdot,\cdot)$ as a
mapping $\bbZ^+\times \cS\mapsto \bbR^+$. Set $W(0,s)=1,\forall
s\in\cS$ as the initial value. Denote $P_{X|S}(\cdot|\cdot)$ as a
mapping $\cX\times\cS\mapsto\bbR^+$ such that for each $s\in\cS$,
$P_{X|S}(\cdot|s)$ is a probability mass function (pmf) on $\cX$,
i.e., $\sum_{x\in\cX}P_{X|S}(x|s)=1$, and $P_{X|S}(x|s)\geq0,\forall
x\in \cX$. The term $W(\cdot,\cdot)$ is the solution to the problem
defined iteratively by:
\begin{equation}\begin{split}\label{eq.dp} &W(n,s)=
\max_{P_{X|S}(\cdot|s)}\min_{s^\prime\in\cS}\left\{W(n-1,s^\prime)\left[\max_{y\in \cY }\sum_{x\in G(y,s^\prime|s)}P_{X|S}(x|s)\right]^{-1}\right\}\\
&\;\;\;\;\;\;\;\;\;\;\;\;\;\;\;\;\;\;\;\;\;\;\;\;\;\;\;\;\;\;\;\;\;\;\;\;\;\;\;\;\;\;\;\;\;\;\;\;\;\;\;\;\;\;
\;\;\;\;\;\;\;\;\;\;\;\;\;\;\;\;\;\;\;\;\;\;\;\;\;\;\;\forall s \in \cS,\text{ and for } n=1,2,3, ...
\end{split}\end{equation}
We adopt the convention that $\frac{1}{0}=\infty$, and, if
$G(y,s^\prime|s)=\emptyset$, $\sum_{x\in
G(y,s^\prime|s)}P_{X|S}(x|s)=0$. One property that can be verified
from the definition and the initial value is that $\forall n\geq 0$,
$\forall s\in \mathcal S$, $W(n,s)\geq 1$.

The main result of this paper is the following theorem:
\begin{theorem}\label{theorem.main}
If $\min_{s\in\cS} V_{|\cS|}(s)>0$,
\begin{equation}\begin{split}\label{eq.c0main}
C_0=& \liminf_{n\rightarrow\infty}\frac{1}{n}\min_{s\in\cS}\log_2W(n,s);
\end{split}\end{equation}

Otherwise $C_0 = 0$.
\end{theorem}
Before proving the theorem, let us verify that the zero-error
capacity of a DMC \cite[Theorem 7]{shannon56} is a special case of
Theorem \ref{theorem.main}. Since a DMC is an FSC with only one
state, $V_{|\cS|}(s)=0$ means that the state is non-{\it positive},
i.e., ``all pairs of input letters are adjacent'', as stated in
\cite[Theorem 7]{shannon56}. If $V_{|\cS|}(s)>0$, for a DMC, define
$M(n)=M(n,s)$ and $G(y)=G(y,s^\prime|s)$.
\begin{eqnarray}
M(n,s)&=& \max_{P_{X|S}(\cdot|s)}\left\{M(n-1)\left[\max_{y\in \cY }\sum_{x\in G(y)}P_{X|S}(x|s)\right]^{-1}\right\}
\nonumber \\
&=&M(n-1)\max_{P_{X}}\left[\max_{y\in \cY }\sum_{x\in
G(y)}P_{X}(x)\right]^{-1},
\end{eqnarray}
and
\begin{eqnarray}
C_0&=& \liminf_{n\rightarrow\infty}\frac{1}{n}\log_2M(n) \nonumber \\
   &=&\log_2\left[\max_{y\in \cY }\sum_{x\in G(y)}P_{X}(x)\right]^{-1},
\end{eqnarray} which is exactly the result for DMC in \cite{shannon56}.

The converse and the direct parts of Theorem \ref{theorem.main} are
proved in Section~\ref{sec.converse} and Section~\ref{sec.direct},
respectively.

\section{Converse}\label{sec.converse}
\begin{theorem}(\it Converse.) \label{theorem.converse} $M(n,s)\leq W(n,s)$, $\forall n=0, 1, 2, ....$ and $\forall s\in\cS$.
\end{theorem}

\begin{proof} We prove the theorem by induction. First, the inequality holds when $n=0$.

Now, suppose $M(k,s)\leq W(k,s)$ is true $\forall k=0,...,n-1$ and
$\forall s\in\cS$. Fix an arbitrary initial state $s_0$. It is
sufficient to show that $M(n,s_0)\leq W(n,s_0)$ to prove the
converse.

For a fixed zero-error code   that has $M(n,s_0)$ messages, we
define
\begin{equation}\begin{split}
u(x|s_0)=&\text{number of messages with first transmitted}\\
         &\quad\text{ letter $x$ when initial state is $s_0$},\\
f(x|s_0)=&\frac{u(x|s_0)}{M(n,s_0)}.
\end{split}\end{equation}

Note that $f(\cdot|s_0)$ is a valid pmf.

After the first transmission, suppose the output is some
\mbox{$y\in\cY$} and the channel goes to state $s_1$. We have
$\sum_{x\in G(y,s_1|s_0)}u(x|s_0)$ messages, each of which with
positive probability gives output $y$ and changes the state to
$s_1$. To guarantee that the decoder can distinguish between these
messages in the following $n-1$ transmission, we must have
$\sum_{x\in G(y,s_1|s_0)}u(x|s_0)\leq M(n-1,s_1)$, which yields
\begin{equation}\begin{split}
&M(n,s_0)\sum_{x\in G(y,s_1|s_0)}f(x|s_0)\leq M(n-1,s_1).
\end{split}\end{equation}

Since the above inequality must hold, $\forall y\in\cY$, and
$\forall s_1\in\cS$
 \begin{equation}\begin{split}
     M(n,s_0)\leq \min_{s_1\in\cS} M(n-1,s_1)\left[\max_{y\in\cY}\sum_{x\in G(y,s_1|s_0)}f(x|s_0)\right]^{-1}
        \end{split}\end{equation}
 %
Since we assumed $M(n-1,s)\leq W(n-1,s)$ for all $s\in\cS$,
\begin{equation}
M(n,s_0) \leq \min_{s_1\in\cS} W(n-1,s_1)
\left[\max_{y\in\cY}\sum_{x\in G(y,s_1|s_0)}f(x|s_0)\right]^{-1}.
\end{equation}
Using the iterative formula of  $W(n,s_0)$ given in  (\ref{eq.dp})
and the fact that $f(\cdot|s_0)$ is a valid pmf, we obtain
\begin{equation}\begin{split}
M(n,s_0)\leq W(n,s_0).
\end{split}\end{equation}
Finally, since $s_0$ is arbitrarily fixed, we have $M(n,s)\leq
W(n,s)$, $\forall s\in\cS$. By induction, the theorem is proved.
\end{proof}

From the converse, Theorem \ref{theorem.converse}, and the
zero-error capacity definition \ref{def_zero_cap}, we have the
following upper bound
\begin{equation}\begin{split}\label{eq.c0fupper}
C_{0}=&\lim_{n\rightarrow\infty}\frac{\min_{s\in\cS}\log_2 M(n,s)}{n}\\
      \leq& \liminf_{n\rightarrow\infty}\frac{\min_{s\in\cS}\log_2 W(n,s)}{n}.
\end{split}\end{equation}

\section{Direct Theorem}\label{sec.direct}
\begin{theorem}\label{theorem.direct}
Assume $\min_{s\in\cS} V_{|\cS|}(s)>0$, then for any initial state
$s\in \mathcal S$  there exists an $n_0>0$ such that for $n>n_0$,
$\lfloor W(n,s)\rfloor$ messages can be transmitted with no more
than $n+|\cS|\lceil\log_2L\rceil$, where $L$  is a positive integer
that does not depend on $n$ and $s$.
\end{theorem}

%

\begin{proof}
The direct part is proved using  deterministic codes
\cite{shannon56} rather than random codes. Let the solution and the
maximizer in the $k$th iteration ($k=1,2,...,n$) of (\ref{eq.dp}) be
$W(k,\cdot)$ and $P_{X|S}^{(k)}(\cdot|\cdot)$, respectively.

Suppose that at the first transmission the channel state is $s_1$
and the total number of messages transmitted through the channel is
$\lfloor W(n,s_1) \rfloor$. We divide the message set into $|\cX|$
groups and transmit $x=i$ for the messages in the $i$th group for
the first transmission. Let $m_i$ denote the number of messages in
the $i$th group. By similar arguments to those in \cite[p.
18]{shannon56}, we can control the size of each group such that:
\begin{equation}\begin{split}\label{eq.direct.m}
 &\text{if } P^{(n)}_{X|S}(i|s_1)>0, \quad\Bigg|\frac{m_i}{\lfloor W(n,s_1)\rfloor}-P^{(n)}_{X|S}(i|s_1)\Bigg|\leq \frac{1}{\lfloor W(n,s_1)\rfloor}; \\
&\text{if } P^{(n)}_{X|S}(i|s_1)=0, \quad m_i=0.
\end{split}\end{equation}
Both the transmitter and the receiver know how the messages are
divided before the transmission. An arbitrary message
$m\in\{1,...,\lfloor W(n,s_1)\rfloor\}$ is selected, and letter $i$
is sent if $m$ belongs to the $i$th group. The number of messages
about which the receiver is uncertain before the first transmission
is $Z_1=\lfloor M(n,s_1)\rfloor$.

After the first transmission, we obtain an output $y_1$, and the
channel state changes to $s_2$. Denote $Z_2$ as the number of
messages that are {\it compatible} with $(y_1,s_2)$, i.e., when
transmitting those messages,  $(y_1,s_2)$ is obtained with positive
probability.   $Z_2$ can be upper bounded in the following way:
\begin{equation}\begin{split}\label{eq.direct.t2}
Z_2 =& \sum_{x\in G(y_1,s_2|s_1)} m_x\\
    =& \lfloor W(n,s_1)\rfloor \sum_{x\in G(y_1,s_2|s_1)} \frac{m_x}{\lfloor W(n,s_1)\rfloor}\\
    \leq& \lfloor W(n,s_1)\rfloor \sum_{x\in G(y_1,s_2|s_1)} \left( P^{(n)}_{X|S}(x|s_1)+\frac{1}{\lfloor W(n,s_1)\rfloor} \right)\\
    \leq & \left\{\lfloor W(n,s_1)\rfloor \max_{y\in\cY}
    \sum_{x\in G(y,s_2|s_1)}  \left( P^{(n)}_{X|S}(x|s_1)\right) \right\} + |\cX|.
\end{split}\end{equation}
For convenience, let us define
\begin{equation}\begin{split}
J^{(k)}(s,s^\prime)= \max_{y\in\cY}\sum_{x\in G(y,s^\prime|s)} P^{(k)}_{X|S}(x|s).
\end{split}\end{equation}
Eq.~(\ref{eq.dp}) and (\ref{eq.direct.t2}) can be written ,
respectively, in terms of $J^{(k)}(s,s^\prime)$ as:
\begin{equation}\label{eq.direct.ineq}
W(k,s)\leq W(k-1,s^\prime) \left[J^{(k)}(s,s^\prime)\right]^{-1},
\forall k\in\bbZ^+,s\in\cS,s^\prime\in\cS.
\end{equation}
%
\begin{equation}\label{e_z2}\begin{split}
 Z_2 \leq & \lfloor W(n,s_1)\rfloor J^{(n)}(s_1,s_2) + |\cX|\\
\leq & W(n-1,s_2)+|\cX|,
\end{split}\end{equation}
where the last inequality is due to (\ref{eq.direct.ineq}).

Since both transmitter and receiver know $s_1$ and $s_2$ and the
transmitter knows the output $y_1$ through feedback, both of them
know which messages are {\it compatible} with $(y_1,s_2)$. In the
second transmission, the transmitter can further divide the
remaining $Z_2$ messages into groups according to
$P_{X|S}^{(n-1)}(\cdot|s_2)$, similar to eq.~{(\ref{eq.direct.m})}.
The way the messages are divided is known to the receiver. Suppose
the output letter is $y_2$ and the state goes to $s_3$. Following
the argument in the previous iteration, we have
\begin{equation}\begin{split}
 Z_3 \leq&  Z_2 J^{(n-1)}(s_2,s_3)+|\cX|\\
     \stackrel{(a)}{\leq}&   W(n-1,s_2)J^{(n-1)}(s_2,s_3)+ |\cX|\left(1+J^{(n-1)}(s_2,s_3)\right)\\
     \stackrel{(b)}{\leq}&   W(n-2,s_3)+ |\cX|\left( 1+J^{(n-1)}(s_2,s_3)
     \right),
\end{split}\end{equation}
where steps (a) and (b) follow from (\ref{e_z2}) and
(\ref{eq.direct.ineq}), respectively.

As the transmission proceeds, the channel state evolves as
$s_1,....,s_n,s_{n+1}$, and the output sequence is $y_1,...,y_n$.
The transmitter divides the remaining uncertain messages according
to $P^{(k)}_X(\cdot|s_k)$ for each transmission. After the $n$th
transmission, the number of messages remaning can be upper bounded
as:
\begin{eqnarray}\label{eq.Tn+1}
 \lefteqn{Z_{n+1}}\nonumber \\
  &\leq&  Z_{n}J^{(1)}(s_{n},s_{n+1}) + |\cX|\nonumber \\
     &\leq&  1+ |\cX|\left( 1+ J^{(1)}(s_{n},s_{n+1})
           + J^{(1)}(s_{n},s_{n+1})J^{(2)}(s_{n-1},s_{n})
 +\cdots +\prod_{i=1}^{n-1} J^{(i)}(s_{n+1-i},s_{n+2-i}) \right) 
\end{eqnarray}
Using Ineq.~(\ref{eq.direct.ineq}) iteratively, we obtain
\begin{equation}\begin{split}
W(k,s_{n+1-k}) \leq \left[\prod_{i=1}^{k}
J^{(i)}(s_{n+1-i},s_{n+2-i})\right]^{-1};
\end{split}\end{equation}
hence we can further upper bound $Z_{n+1}$ as
\begin{equation}
Z_{n+1}\leq 1 + |\cX|\left(1 + \frac{1}{W(1,s_{n})}+
\frac{1}{W(2,s_{n-1})} +\cdots+  \frac{1}{W(n-1,s_{2})}\right).
\end{equation}
Recall the assumption of the theorem $\min_{s\in\cS} V(|\cS|,s)>0$,
which implies, via Theorem \ref{theorem.C=0}, that $C_0>0$, and
follows from Theorem \ref{theorem.converse} we obtain that
\begin{equation}\label{eq.ass}\liminf_{n\rightarrow\infty}\min_{s\in\cS}\frac{1}{n}\log M(n,s)>0.
\end{equation}
Hence, there exists \mbox{$\epsilon>0$} and an integer $n_0$ such
that $\forall s\in\cS$, $\forall n>n_0$, $W(n,s) {\geq} M(n,s) {\geq} 2^{\epsilon n}$ (the first inequality is due to the converse proved in the
previous section, and second inequality is due to (\ref{eq.ass})).
Recall that $M(n,s)\geq 1$; we can thus further upper bound
$Z_{n+1}$ as
\begin{equation}\begin{split}
Z_{n+1}\leq& 1 + |\cX|\left(1+\sum_{k=1}^{n_0}\frac{1}{W(k,s_{n+1-k})}+ \sum_{k=n_0+1}^{\infty} 2^{-\epsilon n}\right)\\
       \leq& 1 + |\cX|\left(n_0+1 + \sum_{k=n_0+1}^{\infty} 2^{-\epsilon n}\right)\\
       =&1 + |\cX|\left(n_0+1 + \frac{2^{-\epsilon (n_0+1)}}{1-2^{-\epsilon}}\right)\\
       \triangleq& L.
\end{split}\end{equation}
Note that $L$ is finite and is independent of $n$ and $s_1$. This
means that after $n$ transmissions, the number of messages about
which the receiver is uncertain is not more than $L$.

The assumption that  $\min_{s\in\cS} V_{|\cS|}(s)>0$ implies that we
can drive the channel to a {\it positive} state with probability 1
in less than $|\cS|$ transmissions. In a positive state, we can
transmit 1 bit of information with zero-error; hence we can now
conclude that there exists a zero-error code such that $\lfloor
W(n,s) \rfloor$ messages can be transmitted with no more than
$n+|\cS|\lceil\log_2L\rceil$ transmissions.
\end{proof}

Based on the direct theorem, it is straightforward to derive a lower bound on the zero-error capacity:
\begin{equation}\begin{split}\label{eq.cu1}
C_{0}\geq& \liminf_{n\rightarrow\infty}\min_{s\in\cS}\frac{\log_2\lfloor W(n,s)\rfloor}{n+|\cS|\lceil\log_2L\rceil}\\
       =& \liminf_{n\rightarrow\infty}\frac{1}{n}\min_{s\in\cS}\log_2W(n,s),
\end{split}\end{equation}
given the condition $\min_{s\in\cS} V_{|\cS|}(s)>0$.
%
Combining ineq.~(\ref{eq.c0fupper}) and ineq.~(\ref{eq.cu1}), we
have proved  eq.~(\ref{eq.c0main}) thus
Theorem~\ref{theorem.main}.

\section{Solving the Dynamic Programming Problem}\label{s_seolving_DP}
Throughout this section, we assume that $\min_{s\in\cS}
V_{|\cS|}(s)>0$, i.e., we focus on channels with positive zero-error
capacity. Let us first introduce a few definitions so that we can
use the standard language of dynamic programming to rewrite
Eq.~(\ref{eq.dp}) in the form of Eq. (\ref{e_nonstandard}).
Basically, we take $\log_2$ on both sides of Eq.~(\ref{eq.dp}).
Define the value function as $J_n(s)=\log_2 W(n,s)$, the action as
$a=P_{X|S}(\cdot|s)$, and the reward as
\begin{equation}
r(s^\prime,a,s)=\log_2\left[\max_{y\in \cY }\sum_{x\in
G(y,s^\prime|s)}P_{X|S}(x|s)\right]^{-1}.\end{equation} And the DP
equation in (\ref{eq.dp}) becomes simply
\begin{equation}\label{e_iteration_DP}\begin{split}
&J_n(s)=\max_{a\in
A}\min_{s^\prime\in\cS}\left\{r(s^\prime,a,s)+J_{n-1}(s^\prime)\right\}.
\end{split}\end{equation}
where $A$ is the action space, $A=\{f(x):\quad \sum_{x}f(x)=1, f(x)\geq0\}$.

Theorem~\ref{theorem.main} states that
\begin{equation}\begin{split}
C_0=& \liminf_{n\rightarrow\infty}\frac{\min_{s\in\cS}J_n(s)}{n}.
\end{split}\end{equation}
Define an operator $T$ as follows,
\begin{equation}\label{eq.T}
(T\circ J)(s)= \max_{a\in A(s)} \min_{s'} \left\{
r(s',a,s)+J(s')\right\}.
\end{equation}
The DP equation can be rewritten in a compact form as follows,
\begin{equation}\begin{split}
&J_n(s)=(T\circ J_{n-1})(s),
\end{split}\end{equation}
with initial value $J_0(s)=0$. We also denote $T^n$ as applying   operator $T$ $n$ times.

\begin{lemma}\label{lemma.Tp} Let $W$ and $V$ denote two functions $\cS\mapsto \bbR^+$. The following properties of $T$ hold:
\begin{itemize}
\item[(a)]  If $W(s)\geq V(s), \; \forall s\in \mathcal S$,   then $T\circ W(s)\geq T\circ V(s)  \; \forall s\in \mathcal S$.
\item[(b)] If $W(s)=V(s)+d\; \forall s\in \mathcal S $, where $d$ is a constant,  then $T\circ W(s)= T\circ V(s)+d, \; \forall s\in \mathcal S$
\end{itemize}
\end{lemma}
\begin{proof}
Both parts of the lemma follow directly from the definition of
$T$.
\end{proof}
\begin{lemma}\label{lemma.Jnp} The following properties of $J_n$ hold:
\begin{itemize}
  \item[(a)] The sequence $\{\min_{s}J_n(s)\}$ is {sup-additive}, i.e.,  
        $\min_{s}J_{n+m}(s) \geq \min_{s}J_n(s)+\min_{s}J_{m}(s)$
  \item[(b)] The sequence $\{\max_{s}J_n(s)\}$ is { sub-additive}, i.e.,  
          $\max_{s}J_{n+m}(s) \leq \max_{s}J_n(s)+\max_{s}J_{m}(s)$
\end{itemize}
\end{lemma}
\begin{proof}
We prove the first property here. The proof of the second one is similar.
\begin{equation}\begin{split}
\min_{s}J_{n+m}(s)&=\min_{s}(T^n \circ J_m)(s)\\
                  &\stackrel{(a)}\geq \min_{s} \left(T^n \circ \min_{s'}J_m(s')\right)(s)\\
                  &= \min_{s} \left(T^n \circ [J_0+\min_{s'}J_m(s')]\right)(s)\\
                  &\stackrel{(b)}{=} \min_{s} (T^n \circ J_0)(s)+
                  \min_{s'}J_m(s'),
\end{split}
\end{equation}
where the steps (a) and (b) follow  from  parts (a) and (b) of Lemma
\ref{lemma.Tp}, respectively.
\end{proof}

\begin{theorem}\label{t_bounds}
The $\liminf$ in Theorem~\ref{theorem.main} can be replaced by
$\lim$, i.e.,
\begin{equation}
C_0 =\lim_{n\rightarrow \infty} \min_{s}\frac{J_n(s)}{n},
\end{equation}
and for all $n\in \mathbb{Z}^+$ the following bounds hold
\begin{equation}\label{eq.bounds}
\min_{s}\frac{J_n(s)}{n}\leq C_0 \leq \max_{s}\frac{J_n(s)}{n}.
\end{equation}
\end{theorem}
\begin{proof}
Following  Lemma \ref{lemma.Jnp} and Fekete's lemma \cite[Ch.
2.6]{CombinatorialOptimization_SchrijverBook}, we obtain the
following two limits:
\begin{eqnarray}
\lim_{n\to\infty} \min_{s}\frac{J_n(s)}{n}&=& \sup_n
\min_{s}\frac{J_n(s)}{n},\nonumber \\
\lim_{n\to\infty} \max_{s}\frac{J_n(s)}{n}&=& \inf_n
\max_{s}\frac{J_n(s)}{n}.
\end{eqnarray}
Finally, from Theorem \ref{theorem.main} we obtain:
\begin{equation}
\max_{s}\frac{J_k(s)}{n} \geq \lim_{n\rightarrow \infty}
\max_{s}\frac{J_n(s)}{n} \geq C_0=\lim_{n\rightarrow \infty}
\min_{s}\frac{J_n(s)}{n}\geq  \min_{s}\frac{J_k(s)}{k},
\end{equation}
 for all $k\in
\mathbb{Z}^+$.
\end{proof}

Eq.~(\ref{eq.bounds}) provides a numerical way to approximate
$C_0$. We now alter to the case that an analytical solution in the
limit can be obtained via Bellman equations.
\begin{theorem}\label{theorem.bellman}
{(\it Bellman equation)} If there exists a positive bounded function
$g: \cS\mapsto \bbR^+$ and a constant $\rho$ that satisfy
\begin{equation}\label{e_bellman}
g(s)+\rho=(T \circ g)(s)
\end{equation}
 then $\lim_{n\rightarrow\infty}
\frac{1}{n} J_n(s)=\rho$.
\end{theorem}
\begin{proof} Assume that there exists a positive bounded function $g: \cS\mapsto \bbR^+$
and a constant $\rho$ that satisfy $g(s)+\rho=(T \circ g)(s)$.
Define $g_0(s)=g(s)$, $g_n(s)=T^ng_0(s)$. Since $ J_0(s)=0\leq g_0(s)$,
then according to part (a) of Lemma \ref{lemma.Tp} $J_n(s)\leq
g_n(s)$. Let $d=\max_s g(s)$. Then $J_0+d\geq g_0$. Hence, according
to part (a) of Lemma \ref{lemma.Tp}, $g_n(s)\leq J_n(s)+d$.
Therefore we have,
\begin{equation}
g_n(s)-d\leq J_n(s)\leq g_n(s).
\end{equation}
Finally, $g(s)+\rho=(T \circ g)(s)$ implies that $\lim_{n\to \infty}
\frac{g_n(s)}{n}=\rho$; hence $\lim_{n\to
\infty}\frac{J_n(s)}{n}=\rho$.
\end{proof}
{\it Remark:} $\rho$ does not depend on the initial state, which
hints that for some decomposable Markov chains, it is impossible
to find a $g: \cS\mapsto \bbR^+$ and a constant $\rho$ to satisfy
the Bellman equation.

\section{Examples}\label{s_examples}
Here we provide three examples and solve them analytically. For the
first two examples, we also find the regular feedback capacity using
\cite{Chen05}.
\begin{example}\label{ex_1}
We consider the very simple example illustrated in
Fig.~\ref{fig.ex1}. The channel has two states. In state 0, the
channel is a binary symmetric channel (BSC) with positive cross
probability. In state 1, the channel is a BSC with  0 cross
probability. Roughly speaking, in state 0, the channel is noisy,
and, in state 1, the channel is noiseless. Suppose the channel state
evolves as a Markov process and is  independent of the input and
output. If the current state is 0, the next channel state is 1 with
certainty. If the state is 1, the channel goes to state 0 with
probability $p>0$ or stays at state 0 with probability $1-p$. Thus,
the channel stays in the noisy state a geometric length of time, and
returns to the perfect state immediately.

\begin{figure}[h]
 \centering
 \psfrag{S0}[][][1]{$S=0$}
 \psfrag{S1}[][][1]{$S=1$}
  \psfrag{0}[][][0.8]{$0$}
 \psfrag{1}[][][0.8]{$1$}

  \psfrag{a}[][][1]{$1$}
   \psfrag{b}[][][1]{$1-p$}
   \psfrag{c}[][][1]{$p$}
   \includegraphics[width=5.5cm]{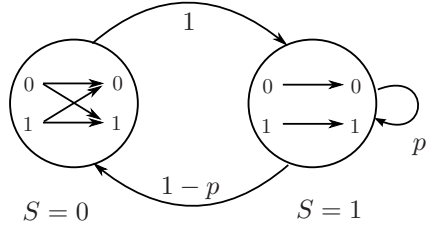}
 \caption{Channel topology of Example 1}
 \label{fig.ex1}
\end{figure}

{\it Finding $C_0$  by calculating $W(n,s)$:} for this channel
$G(y,0|0)=\emptyset$, $G(y,1|0)=\{0,1\}$,
$G(y,0|1)=G(y,1|1)=\{y\}$.
Using eq.~(\ref{eq.dp}) 
, we have the solution to the DP problem of the 1st iteration as
\begin{equation}\begin{split}
W(1,0)=&\max_{P_{X|S}(\cdot|0)}\min\left\{1,1\right\}=1\\
W(1,1)=&\max_{P_{X|S}(\cdot|1)}\Bigg[\max\left\{P_{X|S}(0|1),P_{X|S}(1|1)\right\}\Bigg]^{-1}\\
=&2.
\end{split}\end{equation}
For the 2nd iteration, we have
\begin{equation}\begin{split}
W(2,0)=&\max_{P_{X|S}(\cdot|0)}\Bigg[W(1,1)\min\left\{1,1\right\}\Bigg]=2\\
W(2,1)=&\max_{P_{X|S}(\cdot|1)}W(1,0)\Bigg[\max\left\{P_{X|S}(0|1),P_{X|S}(1|1)\right\}\Bigg]^{-1}=2.
\end{split}\end{equation}
By induction and some simple algebra, we obtain the solution to the DP problem at the $n$th iteration:
\begin{equation}\begin{split}
W(n,0)=&2^{\lfloor n/2\rfloor},\text{ and } W(n,1)=2^{\lceil n/2
\rceil}.
\end{split}\end{equation}
Thus
\begin{equation}\begin{split}
C_{0}=1/2.
\end{split}\end{equation}
Alternatively, we can solve the example by funding a solution to
Bellman equation (\ref{e_bellman}).

{\it Finding $C_0$ via Bellman equation:} the Bellman equation for
the channel is simply the following,
\begin{equation}\begin{split}
g(0)&=g(1)-\rho,\\
g(1)&=1+g(0)-\rho.
\end{split}\end{equation}
Using simple algebra we obtain $\rho=\frac{1}{2}, g(0)=v,
g(1)=v+\frac{1}{2}$. We note that we can achieve the zero-error
capacity with feedback and state information  simply by transmitting
1 bit of information whenever the channel state is 1.

{\it Finding the regular feedback capacity $C^f$:} To calculate the
regular capacity we use the result of Chen and Berger in
\cite[Theorem 6]{Chen05}. The theorem states that if the channel is
strongly irreducible and strongly aperiodic, then the capacity is
\begin{equation}\label{e_chen_berger}
C=\max_{P_{X|S}}\sum_{k=0}^{|\mathcal S|-1} \pi_kI(X;Y|S=k),
\end{equation}
where $\pi_k$ is the equilibrium distribution  of state $k$ induced
by the input distribution $P_{X|S}$.

The channel is strongly irreducible and strongly aperiodic if the
matrix $T$ that is defined as
\begin{equation}
T(k,l)=\min_{x}\{\Pr(S_i=l|X_k=x,S_{i-1}=k)\}
\end{equation}
is irreducible and aperiodic for any $x\in \mathcal X$.
 Since the transition probability of the state does
not depend on the input, and since the state transition matrix is
irreducible and aperiodic for any $p<1$, the capacity is given by
(\ref{e_chen_berger}); hence
\begin{eqnarray}
C(p)&=&\max_{P_{X|S}}\pi_0I(X;Y|S=0)+ \pi_1I(X;Y|S=1)\nonumber \\
&=& \pi_1\nonumber \\
&=& \frac{1}{2-p}
\end{eqnarray}

\begin{figure}[h!]{
 \psfrag{p}[][][0.8]{$p$}
  \psfrag{C(pz)}[][][0.8]{$C(p)$}
\centerline{\includegraphics[width=6.5cm]{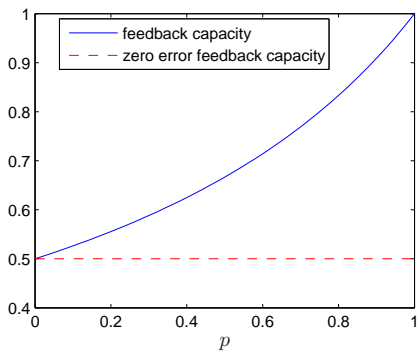}}
\caption{Feedback capacity and zero-error feedback capacity of the
channel in Example 1 for different values of $p=\Pr\{S=1|S=1\}$.}
\label{f_example2} }\end{figure}


\end{example}

\begin{example}\label{ex_2}
Let us consider another channel with two states as illustrated in
Fig.~\ref{fig.ex2}. In state 0, the channel is a Z-channel. In state
1, the channel is a BSC with  0 cross probability. The next channel
state is determined by the output. If the output is 0, the channel
goes to state 0; if the output is 1, the channel goes to state 1;
hence the regular feedback of the output includes the state
information.

 It is tempting to make full use of state 1, i.e., to
transmit 1 bit of information, but as a consequence the channel goes
to the undesirable state 0 half the time, and the rate would be only
$\frac{1}{2}$.

\begin{figure}[h]
 \centering
 \psfrag{S0}[][][1]{$S=0$}
 \psfrag{S1}[][][1]{$S=1$}
  \psfrag{0}[][][0.8]{$0$}
 \psfrag{1}[][][0.8]{$1$}
  \psfrag{a}[][][1]{$y=1$}
   \psfrag{b}[][][1]{$y=0$}
   \psfrag{c}[][][1]{$y=1$}
   \psfrag{d}[][][1]{$y=0$}
   \includegraphics[width=6cm]{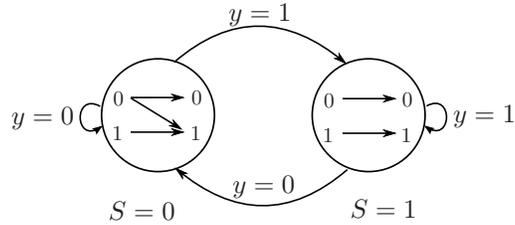}
\caption{Channel topology of Example 2.}
 \label{fig.ex2}
\end{figure}


{\it Finding $C_0$  by calculating $W(n,s)$:} For this channel, $G(0,0|0)=\{0\}$, $G(1,1|0)=\{0,1\}$, $G(0,0|1)=\{0\}$, $G(1,1|1)=\{1\}$ and all the other combinations  yield empty sets. 
For initial state 0, we have
\begin{equation}\begin{split}
W(n,0)
=&\max_{P_{X|S}(\cdot|0)}\min\left\{\frac{W(n-1,0)}{P_{X|S}(0|0)}, W(n-1,1)\right\}\\
=&W(n-1,1)
\end{split}\end{equation}
The maximum is achieved by setting $P_{X|S}(0|0)=0$. For initial state 1, we have
\begin{equation}\begin{split}
W(n,1)
=&\max_{P_{X|S}(\cdot|1)}\min\left\{\frac{W(n-1,0)}{P_{X|S}(0|1)}, \frac{W(n-1,1)}{P_{X|S}(1|1)}\right\}\\
=&\max_{P_{X|S}(\cdot|1)}\min\left\{\frac{W(n-2,1)}{P_{X|S}(0|1)}, \frac{W(n-1,1)}{P_{X|S}(1|1)}\right\}\\
=&W(n-2,1)+W(n-1,1)
\end{split}\end{equation}
By setting \mbox{$P(0|1)=\frac{W(n-2,1)}{W(n-2,1)+W(n-1,1)}$}, the
maximum is achieved. Recall $W(0,1)=1$. Notice that $W(1,1)=2$,
which can be computed directly. Thus, both  $W(n,1)$ and $W(n,0)$
are a Fibonacci sequences (with proper shifts). Therefore, \mbox{$\lim
\frac{\log_2 W(n,1)}{n}=\lim \frac{\log_2 W(n,0)}{n}=\log_2\frac{1+
\sqrt{5}}{2}$}. From Theorem~\ref{theorem.main}, we have
\begin{equation}\begin{split}
C_{0}=\log_2 \frac{1+ \sqrt{5}}{2}\thickapprox 0.6942,
\end{split}\end{equation}
which is the log of the golden ratio. Here, we list the first few
values of $W(n,s)$ in Table~\ref{table.ex2}.

%
%

\begin{table}[h]
\normalsize \centering \caption{$W(n,s)$ which equals to the number
of messages that can be transmitted error-free through the channel
in Example \ref{ex_2} in $n$ steps starting at state
$s$}\label{table.ex2}

\begin{tabular}{||l||c|c|c|c|c||}
\hline \hline \backslashbox{$s$}{$n$} & 1& 2& 3 & 4 &5\\
\hline\hline 0 &1& 2& 3 & 5 &8
 \\ \hline
 1 & 2& 3& 5 & 8 &13\\
    \hline \hline
\end{tabular}
\end{table}

{\it Finding $C_0$ via  a Bellman equation:} Since the channel input is
binary, the actions are equivalent to two numbers:
$p_0=P_{X|S}(0|0)$, $p_1=P_{X|S}(0|1)$. Bellman's equation become
\begin{eqnarray}
J(0)+\rho&=&\max_{0\leq p_0\leq 1}
\min\left\{\log\frac{2}{p_0}+J(0),J(1)\right\} \nonumber \\
J(1)+\rho&=&\max_{0\leq p_1\leq 1}
\min\left\{\log\frac{1}{p_1}+J(0),\log\frac{1}{1-p_1}+J(1)\right\}
\end{eqnarray}
which implies that $p_0=0$ and
\begin{equation}\begin{split}
J(0)&=J(1)-\rho,\\
J(1)&=J(0)+\log_2\frac{1}{p_1}-\rho,\\
\log_2\frac{1}{p_1}+J(0)&=\log_2\frac{1}{1-p_1}+J(1)
\end{split}\end{equation}
the solution of which is
$\rho=\log_2\frac{\sqrt{5}+1}{2},p_1=\frac{3-\sqrt{5}}{2}$.

It is of interest to observe that starting at state 1, any binary
sequence  with length $n$ and no consecutive 0's can be transmitted
with zero-error in $n$ transmissions. The number of such sequences as
a function of $n$ is also a Fibonacci sequence. Since we can always
transmit a 1 to drive the channel from state 0 to state 1, this is
actually one way to achieve the zero-error capacity.

{\it Finding the regular feedback capacity $C^f$:}
%
This channel is not strongly irreducible, since the matrix transition
$P_{S_i|S_{i-1}, X=0}$ is not irreducible; hence, the stationarity of
the optimal policy used by Chen and Berger \cite{Chen05} requires
additional justification. By invoking theory on the infinite-horizon
average-reward dynamic programming  we show that a stationary
policy achieves the optimum of the DP and hence Eq.
(\ref{e_chen_berger}) holds.

The feedback-capacity of the channel in Example 2 can be formulated
according to \cite{Chen05} and \cite{PermuterCuffVanRoyWeissman08}
as:
\begin{equation}
C=\lim_{N \rightarrow \infty} \frac{1}{N}
\max_{\{P_{X_{n}|S_{n}}\}_{n=1}^N}\sum_{n=1}^{N} I(X_n;Y_n|S_{n}),
\end{equation}
and this is equivalent to an infinite-horizon average-reward
DP with finite state space and compact actions
where:
\begin{itemize}
\item the state of the DP is the state of the channels i.e., $S_n$,
\item the actions of the DP are the input distributions
$p_0\in[0,1]$ and  $p_1\in[0,1]$, where $p_0=P_{X|S}(0|0)$,
$p_1=P_{X|S}(0|1)$.
\item the reward at time $n$ given that the state of the DP is 0 or 1 is
$I(X_n;Y_n|S_n=0)=H_b(p_0p)-p_0H_b(p)$ or
$I(X_n;Y_n|S_n=1)=H_b(p_1)$, respectively,
\item the transition probability given the actions $p_1$ and $p_2$ is
$P_{S_n|S_{n-1}}(0|1)=p_1$ and $P_{S_n|S_{n-1}}(0|0)=p_0p$.
\end{itemize}
Next, we claim that it is enough to consider the action $p_1\in
[\epsilon, 1]$ for some $\epsilon>0$. First we note that for
$\epsilon\leq \frac{1}{6}$
\begin{equation}\label{e_entropy_bound}
H(2\epsilon)>H(\epsilon)+\epsilon,
\end{equation}
since $\frac{dH_b(x)}{dx}> 1$ for $x<\frac{1}{3}$.

Next we show that it is never optimal to have an action  $p_1\leq
\frac{1}{6}$. Let $J_n(0)$ and $J_n(1)$ be the maximum rewards to go
in $n$ steps starting at state $0$ and $1$, respectively, and let
assume that the optimal action in state 1 is $p^*_1< \frac{1}{6}$,
then
\begin{eqnarray}
J_n(0)&\stackrel{(a)}{=}&H(p^*_1)+(1-p^*_1) J_{n-1}(0)+ p^*_1 J_{n-1}(1) \nonumber \\
&\stackrel{(b)}{=}&H(p_1)+(1-2 p^*_1) J_{n-1}(0)+ 2p^*_1 J_{n-1}(1) +p^*_1(J_{n-1}(0)-J_{n-1}(1))\nonumber \\
&\stackrel{(c)}{\leq }&H(p^*_1)+(1-2 p^*_1) J_{n-1}(0)+ 2p^*_1 J_{n-1}(1) +p^*_1\nonumber \\
&\stackrel{(d)}{< }&H(2p_1)+(1-2 p^*_1) J_{n-1}(0)+ 2p^*_1
J_{n-1}(1),
\end{eqnarray}
where step (a) follows from the dynamic programming formulation;
step (b) follows from the fact that we added and subtracted
$p^*_1(J_{n-1}(0)-J_{n-1}(1))$; and step (c) follows from the fact that
$J_{n-1}(0)-J_{n-1}(1)\leq 1$; this is because we can choose
$p_0=0$, which means that in one epoch time we can cause the state
to change from 0 to 1 with probability 1, and the reward in one
epoch time is always less than 1. Finally, step (d) follows from
(\ref{e_entropy_bound}). Since step (d) corresponds to the action
$2p^*_1$,  it implies that an optimal policy would never include the
action $p^*_1< \frac{1}{6}$.

Now we invoke \cite[Theorem 4.5]{Arapos93_average_cose_survey} that
states that if the reward is a continuous function of the actions,
and for any action the corresponding state chain is irreducible
(unchain), then the optimal policy is stationary. Since the reward
function is continuous in $p_0,p_1$ and since for any
$p_0\in[0,1],p_1\in [\frac{1}{6},1]$ the state process is a
irreducible, we conclude that the optimal policy $p^*_1,p^*_2$ is
stationary (time-invariant), and therefore the capacity is given by
(\ref{e_chen_berger}).


\begin{figure}[h!]{
 \psfrag{pz}[][][0.8]{$p$}
  \psfrag{C(pz)}[][][0.8]{$C(p)$}
\centerline{\includegraphics[width=6.5cm]{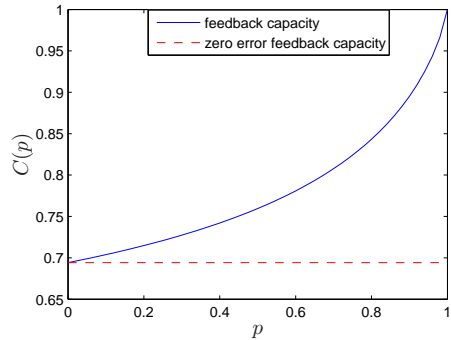}}
\caption{Capacity and zero-error capacity of the channel in Example
2 for different values of $p=\Pr\{Y=0|X=0,S=0\}$.}
\label{f_example2} }\end{figure}

 Now, using (\ref{e_chen_berger}), we obtain that the
regular feedback capacity as a function of $p$ is
\begin{equation}\label{e_capacity_example2}
C^f(p)=\max_{p_0,p_1} ( \pi_0\left(H_b(p_0p)-p_0H_b(p)\right)+
\pi_1H_b(p_1)),
\end{equation}
where $(\pi_0,\pi_1)$ are the equilibrium distributions given by
$\pi_0=\frac{p_1}{1+p_1-p_0p}$ and $\pi_1=1-\pi_0$. Fig.
\ref{f_example2} shows a numerical evaluation 
(\ref{e_capacity_example2}) as a function of $p$.

\end{example}

\begin{example}
We consider here an example with three states with a trinary input
and trinary output. The topology of the channel is depicted in Fig.
\ref{f_example3}. The channel conditional distribution $P(s',y|x,s)$
has  the form of $P(s',y|x,s)=P(s'|x,s)P(y|x,s)$, where state $s=0$
is a perfect state , $s=1$ is a good state and $s=0$ is a bad state;
the states 1,2,3 can transmit $\log3, 1$ and $0$ bits with zero
error probability.

We first evaluate the zero-error capacity numerically using the
dynamic programming value iteration, i.e., Eq. (\ref{eq.bounds}),
and then, using the numerical evaluation, we conjecture an
analytical solution, which we verify via the Bellman equation.

\begin{figure}[h!]{
\psfrag{X}[][][0.8]{$X$} \psfrag{Y}[][][0.8]{$Y$}

\centerline{\includegraphics[width=8cm]{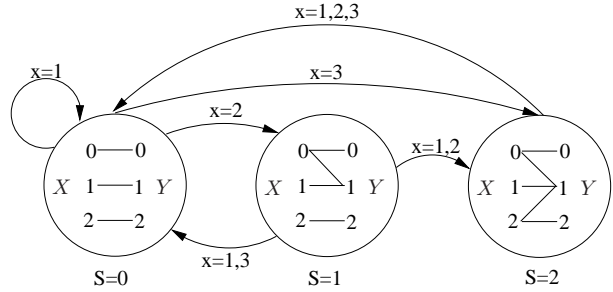}}
\caption{Channel topology of Example 3.} \label{f_example3}
}\end{figure}

{\it Evaluating $C_0$ using a value iteration algorithm:} We
calculated  50 iterations of the DP value iteration formula given in
(\ref{e_iteration_DP}). The action space of player 1 is the
stochastic matrix $P_{X|S}$, and we quantize each element of the
stochastic matrix with a $10^{-4}$ resolution. Fig. \ref{f_J}
depicts the value of $\max_s J_n(s)$ and $\max_s J_n(s)$ which
according to Theorem \ref{t_bounds} are upper and lower bounds, respectively, on
the zero-error capacity.

After 50 iterations, we obtain that the first player's action $P_{X|S}$
is given by
\begin{equation} \label{e_PXS}
P_{X|S} =\left[ \begin{tabular}{l l l}
0.4656 & 0.3177 & 0.2167 \\
0 & 0.3177 & 0.6823 \\
0 & 0 & 1 \\
\end{tabular}
\right],
\end{equation}
and the the reward $J_{50}(s)-J_{49}(s)$, which is an estimate of the
zero-error capacity, is $1.10283$ for all $s\in{0,1,2}$.
\begin{figure}[h!]{
\psfrag{max}[][][0.9]{$ \leftarrow \max_s
\frac{J_n(s)}{n}$} \psfrag{min}[][][0.9]{$\color{red} \;\;
\leftarrow \min_s \frac{J_n(s)}{n}$} \psfrag{It}[][][0.9]{$n$}
\psfrag{C}[][][0.9]{$\;\;\;\;\;\;\;\;\;\;\;\;\;\;\;\;\;\;\;\;\;\;\;\;\;\;\;\;\;\;\;\;\;\;\;\;\;
\leftarrow J_{50}(s)-J_{49}(s)=1.1028$}
\centerline{\includegraphics[width=6.5cm]{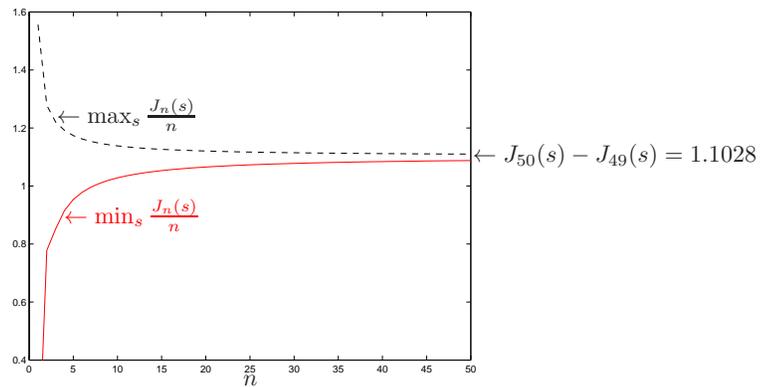}}
\caption{Upper bound, $\max_s J_n(s)$, and lower bound, $\min_s
J_n(s)$, on the zero-error feedback capacity of the channel in
Example 3. The value $J_{50}(s)-J_{49}(s)=1.102$ is an estimate of
$C_0$.} \label{f_J} }\end{figure}

{\it Analytical solution via Bellman equation:} We conjecture that
the optimal policy of Player 1 is a stochastic matrix of the form
given in (\ref{e_PXS}), i.e., $P_{X|S}(1|1)=P_{X|S}(1|0)$, and
$P_{X|S}(0|1)=P_{X|S}(0|2)=P_{X|S}(1|2)=0$. Based on this
assumptions and the notation $a_0\triangleq P_{X|S}(0|0)$ and
$a_1\triangleq P_{X|S}(1|0)$, the Bellman equation becomes:
\begin{eqnarray}
\rho+J(0)&=&\max_{a_0,a_1}\min\{-\log a_0+J(0), -\log a_1+J(1),-log
(1-a_1-a_0)+J(2) \}\nonumber \\
\rho+J(1)&=&\max_{a_1}\min\{-\log a_1+J(2), -\log (1-a_1)+J(0) \}\nonumber \\
\rho+J(2)&=&J(0).
\end{eqnarray}
Using simple algebraic manipulation, we obtain that
\begin{eqnarray}
a_1&=&(1-a_1)^3\nonumber \\
\rho&=&\log \frac{(1-a_1)}{a_1},
\end{eqnarray}
which implies that $a_1=1+u-\frac{1}{3u}$, where
$u=\sqrt[3]{-\frac{1}{2}+\sqrt{\frac{1}{4}+\frac{1}{27}}}$, hence
$a_1=0.31767...$ and
\begin{equation}
C_0=-\log (1-a_1)=1.102926....
\end{equation}

\end{example}
\section{Conclusions}
We introduced a DP formulation for computing the zero-error feedback
capacity for FSCs with state information at the decoder and encoder.
The DP formulation, which can also be viewed as a stochastic game
between two players, is a powerful tool that allows us to evaluate
numerically the zero-error feedback capacity and in many cases as
shown in the paper, to find an analytical solution via a fixed-point
 equation.

\section{Acknowledgements}
The authors would like to thank Professor Thomas Cover for very helpful discussions and comments.
This work is supported by the National Science Foundation through the grants CCF-0515303 and CCF-0635318.

\bibliographystyle{unsrt}
\bibliographystyle{IEEEtran}
\bibliography{ref}

\begin{thebibliography}{10}

\bibitem{shannon56}
C.~E. Shannon.
\newblock The zero error capacity of a noisy channel.
\newblock {\em IEEE Trans. Inf. Theory}, IT-2:8--19, 1956.

\bibitem{KornerOrlitski98Zero-error}
J.~K\"orner and A.~Orlitsky.
\newblock Zero-error information theory.
\newblock {\em IEEE Trans. Inf. Theory}, 44(6):2207--2229, 1998.

\bibitem{Chen05}
J.~Chen and T.~Berger.
\newblock The capacity of finite-state {M}arkov channels with feedback.
\newblock {\em IEEE Trans. Inf. Theory}, 51:780--789, 2005.

\bibitem{Ahl_kaspi87}
R.~Ahlswede and A.~Kaspi.
\newblock Optimal coding strategies for certain permuting channels.
\newblock {\em IEEE Trans. Inf. Theory}, 33(3):310--314, 1987.

\bibitem{PermuterCuffVanRoyWeissman07_Chemical}
B.~Van~Roy H.~Permuter, P.~Cuff and T.~Weissman.
\newblock Capacity and zero-error capacity of the chemical channel with
  feedback.
\newblock In {\em Proc. International Symposium on Information Theory (ISIT)},
  France, Nice, 2007.

\bibitem{NayaRose05}
J.~Nayak and K.~Rose.
\newblock Graph capacities and zero-error transmission over compound channels.
\newblock {\em IEEE Trans. Inf. Theory}, 51(12):4374--4378, 2005.

\bibitem{Gallager68}
R.~G. Gallager.
\newblock {\em Information theory and reliable communication}.
\newblock Wiley, New York, 1968.

\bibitem{CombinatorialOptimization_SchrijverBook}
A.~Schrijver.
\newblock {\em Combinatorial Optimization - Polyhedra and Efficiency}.
\newblock Springer, 2003.

\bibitem{Arapos93_average_cose_survey}
A.~Arapostathis, V.~S. Borkar, E.~Fernandez-Gaucherand, M.~K. Ghosh, and
  S.~Marcus.
\newblock Discrete time controlled {M}arkov processes with average cost
  criterion - a survey.
\newblock {\em SIAM Journal of Control and Optimization}, 31(2):282--344, 1993.

\bibitem{Bertsekas05}
D.~P. Bertsekas.
\newblock {\em Dynamic Programming and Optimal Control: Vols 1 and 2}.
\newblock Athena Scientific, Belmont, MA., 3 edition, 2005.

\bibitem{Shapley1953StochasticGame}
L.S. Shapley.
\newblock Stochastic games.
\newblock {\em Proceedings of the National Academy of Sciences}, 39:1095–1100,
  1953.

\bibitem{Filar97CompetitveMDP}
J.~Filar and K.~Vrieze.
\newblock {\em Competitive Markov Decision Processes}.
\newblock Springer, New York and Heidelberg, 1997.

\bibitem{PermuterCuffVanRoyWeissman08}
H.~H. Permuter, P.~Cuff, B.~Van Roy, and T.~Weissman.
\newblock Capacity of the trapdoor channel with feedback.
\newblock {\em IEEE Trans. Inf. Theory}, 54(7):3150--3165, 2009.

\end{thebibliography}
\end{document}